\documentclass{article}

\usepackage{times}
\usepackage{shortbold, dsfont}  % Amos' short bold type and sans serif
\usepackage{verbatim}           % comment

\usepackage[top=1in, bottom=1in, left=1.2in, right=1.2in]{geometry}
\usepackage{color}

\usepackage{graphicx, subfig}
\usepackage{tikz, pgfgantt}
\usetikzlibrary{arrows,backgrounds,positioning,fit,calc,intersections}

\usepackage[longnamesfirst]{natbib}
\usepackage{amsmath, amssymb, amsthm}
\usepackage{algorithm, algorithmic}
\usepackage{paralist}       % provide inline list
\usepackage[raggedright]{titlesec}
\usepackage{authblk}

\usepackage[unicode, backref=page, bookmarks=false]{hyperref}
\usepackage{url}

\hypersetup{
    colorlinks=true,
    linkcolor=[rgb]{0.5,0,0},
    citecolor=[rgb]{0,0,0.5},
     urlcolor=[rgb]{0,0.5,0},
}

\usepackage{caption}    % caption title

\newtheorem{proposition}{Proposition}
\newcommand{\Real}{\mathds{R}}
\newcommand{\dif}{\mathrm{d}}
\newcommand{\ind}{\mathds{1}}

\DeclareMathOperator*{\argmin}{arg\,min}
\DeclareMathOperator*{\E}{\mathds{E}}
\DeclareMathOperator*{\spn}{\operatorname{span}}

\numberwithin{equation}{section}

\title{Multi-period Trading Prediction Markets with Connections to Machine Learning}
\author{Jinli Hu}
\author{Amos Storkey}
\affil{Institute for Adaptive and Neural Computation \\ School of Informatics, The University of Edinburgh \\ 10 Crichton Street, Edinburgh, EH8 9AB \\ \texttt{\{J.HU, A.STORKEY\}@ED.AC.UK}}

\date{March 2014}
\begin{document}

\maketitle

\begin{abstract}
    We present a new model for prediction markets, in which we use risk measures to model agents and introduce a market maker to describe the trading process. This specific choice on modelling tools brings us mathematical convenience. The analysis shows that the whole market effectively approaches a global objective, despite that the market is designed such that each agent only cares about its own goal. Additionally, the market dynamics provides a sensible algorithm for optimising the global objective. An intimate connection between machine learning and our markets is thus established, such that we could
    \begin{inparaenum}[1)]
        \item analyse a market by applying machine learning methods to the global objective, and
        \item solve machine learning problems by setting up and running certain markets.
    \end{inparaenum}
\end{abstract}

\section{Introduction}
Following the mainstream interest in ``big data'', one valuable direction of machine learning is towards to building up distributed, scalable and self-incentivised systems which could organise for solving large scale problems. Recently, prediction markets \citep{wolfers2004} show the promise of being the abstract framework for machine learners to design these systems. As one type of markets, prediction markets naturally introduce the concepts such as self-incentivised computation and distributed environment. Additionally, the close relationship between prediction markets and probabilities shed light on a new way of achieving probabilistic modelling \citep{storkey2011}.

Since \citet{pennock1996toward}, researchers have spent decades on building connections between machine learning and prediction markets. However, this problem has still not been well solved. One reason is that the framework of prediction markets is somehow too flexible, and in order to analyse the markets for machine learning goals one has to first specify a market model to describe the prediction markets. The other reason is that given a market model, we may still not know what the market is doing, even if we understand agent behaviours and market mechanisms. As distinct from most machine learning methods which explicitly define and optimise certain objectives, markets only introduce local objectives to each individual agent. To interpret a market as a machine learning method, we have to find the global objective that the market aims to optimise. This idea motivates our work.

Instead of just focusing on market mechanisms \citep{chen2010}, we would like to incorporate the agents and analyse our market as a whole. This setting is similar to \citep{storkey2011,penna2012,barbu2012}; but unlike \citet{barbu2012}, we will build a model on agent behaviours; and unlike \citet{storkey2011} and \citet{penna2012}, we model agents using risk measures, which makes our markets analytical.

The novel results of this paper are:
\begin{itemize}
  \item establishing a new prediction market model which both inherits the strengths of prediction markets and has mathematical convenience;
  \item giving explicitly the global objective that the market aims to optimise as a whole, and interpreting the market trading process as a sequential optimisation procedure of the global objective;
  \item strengthening the intimate connections between machine learning and markets by showing that the market effectively solves the dual of certain machine learning problems.
\end{itemize}

\section{A General Prediction Market Setup}\label{sec:setup}
Let $\Omega$ be the space of all possible future states. We say a prediction market is built on $\Omega$ if it trades securities associated with the future state $\omega\in\Omega$. Specifically, securities are defined as a set of random variables $\{\xi_k(\cdot)\} = \{\xi_1(\cdot),\xi_2(\cdot)\ldots,\xi_K(\cdot)\}$. Each $\xi_k(\cdot):\Omega\to\Real$ is a payment function, that is, one unit of this security will pay to the holder $\xi_k(\omega)$ if $\omega$ turns out to be the future state. This definition is quite general, and securities defined in this way are also referred to as \emph{complex securities} \citep{abernethy2013}. We require that all securities $\{\xi_k(\cdot)\}$ (collected into the vector $\xiB(\cdot)$) are linearly independent, that is, for $\aB\in\Real^K$, we have $\aB\cdot\xiB(\cdot) = 0$ only if $\aB = \zeroB$. If they are not, then we can always pick a subset $\{\xi_{k'}(\cdot)\}$ of linearly independent securities from $\{\xi_k(\cdot)\}$ such that all the other securities in $\{\xi_k(\cdot)\}$ can be represented by the linear combination of $\{\xi_{k'}(\cdot)\}$ \citep{kreyszig2007}. Therefore it is redundant to consider $\{\xi_k(\cdot)\}$ that are not linearly independent. As an example, the Arrow-Debreu security is a special case of complex securities. When the sample space $\Omega$ is discrete and contains only finite number of states, Arrow-Debreu securities are a set of $K=|\Omega|$ securities, in which the $k$-th one pays one unit if the $k$-th state is true: $\xi_k(\omega) = \ind(\omega=k)$. Note that in general cases $K<|\Omega|$, e.g.\ when the value of $\omega$ is continuous, there will be infinite number of states but we always have a finite $K$ for practice.

Agents can only trade these predefined securities. The behaviour of an agent is characterised by its \emph{portfolio} $\{w,s_k\} = \{w,s_1,s_2,\ldots,s_K\}$, where $w$ is the amount of money that the agent has, and $s_k$ is the amount of shares the agent holds in security $k$. We collect all $s_k$ into vector $\sB$. If an agent has a portfolio $\{w,s_k\}$, the total payment of the securities is
\begin{equation}\label{eq:riksypayment}
    X(\cdot) = \sB\cdot\xiB(\cdot),
\end{equation}
where $X(\cdot):\Omega\to\Real$ is in essence a random variable on $\Omega$. We call $X(\cdot)$ the \emph{risky asset} because of its uncertainty and $w$ the \emph{risk-free asset}. The gross payment is thus
\begin{equation}
    \hat{X}(\cdot) = w + \sB\cdot\xiB(\cdot) = w + X(\cdot),
\end{equation}
which is also a random variable. We call $\hat{X}(\cdot)$ the \emph{(gross) asset}. Denote $\mathcal{X}$ the set of all $X(\cdot)$ that are accessible for an agent, and similarly $\mathcal{\hat{X}}$ the set of all $\hat{X}(\cdot)$. Notice that since $\{\xi_k(\cdot)\}$ are linearly independent, there exists a unique map (bijection) between $X(\cdot)$ and $\sB$ via (\ref{eq:riksypayment}). Therefore a portfolio could also be represented by $\{w,X(\cdot)\}$. In our setting $\mathcal{X}\subseteq\spn(\xi_1(\cdot),\xi_2(\cdot)\ldots,\xi_K(\cdot))$, but it is possible to make $\mathcal{X}$ more abstract, which is not a space spanned by a prefixed number of securities but allows new security types to be added on the fly. This discussion is beyond our scope.

A market consists of two processes, that
\begin{inparaenum}[1)]
    \item each agent chooses a portfolio $\{w,X(\cdot)\}$ it would like to hold, and
    \item agents try to move to their preferred portfolio by trading.
\end{inparaenum}
To describe the decision making process we need a model of portfolio selection, while to describe trading we need to specify a market mechanism. These two parts will be discussed in Section \ref{sec:preferences} and \ref{sec:trading}.

Later in this paper, when the context is clear we will omit parentheses and write a random variable in an uppercase letter, e.g.\ $X$ (except the securities, which are denoted by $\xi$), and use the lowercase of the same letter for the value of it, e.g.\ $x$. We will also write functionals in letters without any parentheses.

\section{Preferences on Assets}\label{sec:preferences}
Agents select assets based on their preferences. An agent's preference order of two assets is measured by a functional $f:\mathcal{\hat{X}}\to\Real$, such that the agent prefers one asset $\hat{X}$ than the other asset $\hat{Y}$ if and only if $f(\hat{X})>f(\hat{Y})$, and that the agent is indifferent between $X$ and $Y$ if and only if $f(\hat{X})=f(\hat{Y})$. There are plenty of theories on choosing and analysing a specific form of $f$. These includes expected utility theory (EUT) \citep{vonneumann2007}, dual utility theory \citep{yaari1987}, risk measures \citep{artzner1999}, etc.\ EUT is perhaps the most popular theories in economics and game theory, while risk measures are commonly seen in finance literature. We choose to use risk measures to model agent behaviours. We introduce risk measures in this section, while putting the detailed justification of using risk measures and its relation to EUT in Section \ref{sec:relatedwork}.

\subsection{Risk measures}
As is indicated by their name, risk measures assign higher scores to assets that are more ``risky''. They can also be understood as measures of the potential loss of choosing certain asset. A \emph{(monetary) risk measure} is defined as a functional $\rho:\mathcal{X}\to\Real$ such that $\rho(0)$ is finite and $\rho$ satisfies the following conditions \citep{artzner1999}:
\begin{description}
    \item[Translation invariance] If $X\in\mathcal{X}$ and $m\in\Real$, then
        \begin{equation}\label{eq:translation}
            \rho(X + m) = \rho(X) - m.
        \end{equation}
    \item[Monotonicity] If $X,Y\in\mathcal{X}$ and $X\leq Y$, then
        \begin{equation}
            \rho(X)\geq\rho(Y).
        \end{equation}
\end{description}
Here $X\leq Y$ should be understood as $P(x\leq y) = 1$, that is, with the probability of one that $X$ will generate a lower return than $Y$. Thus monotonicity indicates that an asset with a better return deserves a lower risk. Due to translation invariance, a risk measure maps any risk-free asset to itself, and is additive w.r.t.\ any risk-free asset. Therefore, the output of a risk measure has the same unit with a risk-free asset, and can be calculated like an asset.

The domains of risk measures and the preference functional $f$ are different, as risk measures are defined on $\mathcal{X}$ while the space of assets that agent can hold is $\mathcal{\hat{X}}$. Fortunately, we could easily extend the definition to the domain $\mathcal{\hat{X}}$ by applying translation invariance (\ref{eq:translation})
\begin{equation}
    \rho(\hat{X}) = \rho(X + w) = \rho(X) - w, \quad \forall \hat{X}\in\mathcal{\hat{X}}.
\end{equation}
A corresponding $f$ can thus be obtained by $f = -\rho$.

Risk measures are very generic. In our discussion we will use both risk measures and a specific class of them, the \emph{convex} risk measures \citep{follmer2002}. A risk measure is convex if $\forall X_1,X_2\in\mathcal{X}$ and $\lambda\in[0,1]$
\begin{equation}\label{eq:convexity}
    \rho(\lambda X_1 + (1-\lambda)X_2) \leq \lambda\rho(X_1) + (1-\lambda)\rho(X_2).
\end{equation}
It says that the risk of a combination of two assets should not be higher than holding them separately. In other words, convex risk measures encourage diversification, which is a natural condition on risk measures.

\begin{comment}
The coherent risk measure \citep{artzner1999} is a special convex risk measure which satisfies an extra property called \emph{positive homogeneity}:
\begin{equation}
    \rho(\lambda X) = \lambda\rho(X), \quad \forall X\in\mathcal{X}, \lambda\geq0.
\end{equation}
Positive homogeneity says that the risk of holding a double amount of the risky asset will be simply doubled.
\end{comment}

\subsubsection{Examples of risk measures}
A famous non-convex risk measure is the \emph{Value at Risk} (VaR) \citep{linsmeier2000}, which outputs a threshold loss $l$ such that the probability of $-X$ exceeding $l$ is smaller than a predefined level
\begin{equation}
    \mathrm{VaR}_\alpha(X) \equiv \inf\{l\in\Real\mid P(-X>l)\leq 1-\alpha\}.
\end{equation}
A famous convex risk measure is the \emph{Entropic risk measure} \citep{follmer2004}
\begin{equation}\label{eq:entropic}
    \rho_E = \frac{1}{\theta}\log M_X(-\theta) = \sup_{Q\in\mathcal{P}}\E\nolimits_Q[-X] - \frac{1}{\theta}D[Q||P].
\end{equation}
Here $M_X(t)\equiv \E_P[e^{tX}]$ is the moment-generating function, and $D[\cdot||\cdot]$ is the KL-divergence (and this is where ``entropic'' comes in). We mention that the second representation of $\rho_E$ holds for all convex risk measures, and this representation becomes the key to connecting the markets to machine learning (cf.\ Section \ref{sec:objective}).

\subsection{Rational Choices}
Recall that a portfolio that leads to a higher value of $f(\hat{X})$ is preferred. Thus the favourite portfolio of an agent should be the one that maximises $f$, which we denote by $\{w,X\}^{opt}$. The behaviour of choosing $\{w,X\}^{opt}$ is called the \emph{rational choice}, and an agent is rational if it always chooses $\{w,X\}^{opt}$ as its trading goal. Since in our framework $f = -\rho$, a rational agent will choose will $\{w,X\}^{opt}$ under the rule of
\begin{equation}\label{eq:agentgoal}
    \min_{\{w,X\}} \rho(\hat{X}) = \min_{\{w,X\}} \rho(w + X).
\end{equation}

In a market an agent only cares about its own goal (\ref{eq:agentgoal}). It seems like this property prevents us from linking markets to machine learning methods, as the latter always aim to achieve certain global objectives. However, with a careful design, we can let our markets implicitly define global objectives and make an agent contributes to the global objective at the same time as it achieves its own goal.

\section{Multi-period Trading Markets}\label{sec:trading}
In this section we will build our market, a multi-period trading market whose trades are driven by a market maker. ``Multi-period'' is used to indicate that the prices of the securities are allowed to vary at different time steps, and that agents can trade with the market maker in multiple times \citep{follmer2004}. The market maker is introduced to simplify the market mechanism and to make the market run efficiently.

It is difficult to characterise the trading process in the markets with unspecified mechanisms, and those markets may not run efficiently. For example, there may not exist a consistent agreement among agents on how much should be paid to buy/sell one share of a security. Moreover, one agent who wants to sell a certain amount of shares may not find any buyers \citep{chen2007}. One way to simplify the trading process is by introducing a market maker \citep{hanson2007}. A market maker is a special agent. It is a price maker, who defines the price for trading each security. All agents are only allowed to trade with the market maker. They can, however, make a trade at any time as long as they agree to pay under the market maker's pricing. The pricing rule of a market maker at time step $t$ is a functional $c_t:\mathcal{X}\to\Real$. At different time steps the cost for purchasing an asset may be different, i.e.\ it may happen that $c_t(X)\neq c_{t'}(X)$ when $t\neq t'$.

Suppose that an agent has a portfolio $\{w_{t-1},X_{t-1}\}$ at time $t-1$ and it would like to buy $\Delta X_t$ from the market maker at $t$. The agent cannot propose an arbitrary price $\Delta w_t$ for $\Delta X_t$ but has to accept the price provided by the market maker $\Delta w_t = -c_t(\Delta X_t)$. The updated portfolio is thus restricted to $\{w_{t-1}-c_t(\Delta X_t),X_{t-1}+\Delta X_t\}$, and the updated asset is restricted to $\hat{X}_t = X_{t-1}+w_{t-1}+\Delta X_t-c_t(\Delta X_t)$. Now a rational agent only cares about choosing its optimal purchase amount $\Delta X_t$ such that $\rho(\hat{X}_t)$ is minimised:
\begin{equation}\label{eq:singlemin}
    \min_{\{w_t,X_t\}} \rho(\hat{X}_t) = \min_{\Delta X_t\in\mathcal{X}} \rho(X_{t-1}+\Delta X_t+w_{t-1}-c_t(\Delta X_t)).
\end{equation}
This portfolio selection procedure leads to Algorithm \ref{alg:select}.

\begin{algorithm}
    \caption{$\text{\bfseries Select}(\{w_{t-1},X_{t-1}\},\rho(\cdot),c_t(\cdot))$: portfolio selection of a rational agent}
    \label{alg:select}
    \begin{algorithmic}
        \REQUIRE portfolio $\{w_{t-1},X_{t-1}\}$, risk measure $\rho(\cdot)$, pricing rule $c_t(\cdot)$
        \STATE compute $\Delta X_t = \argmin_{\Delta X'_t\in\mathcal{X}} \rho(X_{t-1}+\Delta X'_t+w_{t-1}-c_t(\Delta X'_t))$ using (\ref{eq:singlemin})
        \ENSURE $\{\Delta X_t,-c_t(\Delta X_t)\}$
    \end{algorithmic}
\end{algorithm}

We now consider a multi-period market which involves a set $A = \{1,2,\ldots,N\}$ of agents and a market maker. Assume that at each round $t$ there is only one agent $a_t\in A$ that trades with the market maker. This assumption indicates that each agent trades with the market maker separately, and they do not cooperate to make a joint purchase. $\{a_1,a_2,\ldots,a_T\}$ is thus the trading queue of the market. Since there are multiple agents, we use an extra subscript to distinguish the portfolios of different agents. For example, an agent $n\in A$'s portfolio at time $t$ is $\{w_{n,t}, X_{n,t}\}$. The initial values are denoted with the subscript $t=0$. We collect all $w_{n,t},X_{n,t},\hat{X}_{n,t}$ into vectors $\wB_t,\XB_t$ and $\hat{\XB}_t$, respectively. We assume that agents do not bring in any risky asset at the beginning, which is a natural assumption since only the market maker can issue securities. This assumption means we have $\XB_0 = \zeroB$ and so $\hat{\XB}_0 = \wB_0$.

At time $t$, only the agent $a_t$ updates its portfolio by trading with the market maker while all the other agents keep the same portfolios as at $t-1$. Suppose the asset that the agent $a_t$ would like to purchase is $\Delta X_{a_t,t}$, then for all $n\in A$
\begin{subequations}\label{eq:update}
    \begin{align}
        X_{n,t} &= X_{n,t-1} + \ind(n=a_t)\Delta X_{a_t,t}, \\
        w_{n,t} &= w_{n,t-1} - \ind(n=a_t)c_t(\Delta X_{a_t,t}).
    \end{align}
\end{subequations}
Algorithm \ref{alg:multiple} runs a multi-period trading market.
\begin{algorithm}
    \caption{A multi-period market with a set $A$ of agents and a market maker}
    \label{alg:multiple}
    \begin{algorithmic}
        \REQUIRE initial portfolios $\{\wB,\XB_0\}$, risk measures $\{\rho_n(\cdot)\}$, pricing rule $c_t(\cdot)$, time period $T$
        \FOR{$t=1$ \TO $T$}
            \FOR{each agent $n\in A$}
                \STATE propose $\{\Delta X_{n,t}, -c_t(\Delta X_t)\} = \text{\bfseries Select}(\{w_{n,t-1},X_{n,t-1}\},\rho_n(\cdot),c_t(\cdot))$ using Algorithm \ref{alg:select}
            \ENDFOR
            \STATE trade happens between the market maker and $a_t$
            \FOR{each agent $n\in A$}
                \STATE update their portfolios using (\ref{eq:update})
            \ENDFOR
        \ENDFOR
    \end{algorithmic}
\end{algorithm}

We can also split Algorithm \ref{alg:multiple} into two routines, in terms of the market maker and each agent, respectively (Algorithm \ref{alg:marketmaker} and \ref{alg:agent}). We say this to emphasise the fact that each agent in the market has its \emph{own objective} (which is to achieve the optimal portfolio based on its unique preferences), plus a communication with the market maker.

\begin{algorithm}
    \caption{The market maker in a multi-period market}
    \label{alg:marketmaker}
    \begin{algorithmic}
        \REQUIRE time period $T$
        \FOR{$t=1$ \TO $T$}
            \STATE publish a pricing rule $c_t(\cdot)$
            \STATE collect trading request $\{\Delta X_{n,t}\}$ from agents
            \STATE choose agent $a_t$ and make trade, $\Delta X_t \equiv \Delta X_{a_t,t}$
            \STATE update $c_t(\cdot) \to c_{t-1}(\cdot)$
        \ENDFOR
        \STATE close the market
        \ENSURE $\{a_t\},\{\Delta X_t\}$
    \end{algorithmic}
\end{algorithm}

\begin{algorithm}
    \caption{An agent $n\in A$ in a multi-period market}
    \label{alg:agent}
    \begin{algorithmic}
        \REQUIRE initial portfolio $\{w_{n,0},X_{n,0}\}$, risk measure $\rho_n(\cdot)$, starting point $t=1$
        \REPEAT
            \STATE receive the pricing rule for time $t$
            \STATE compute $\{\Delta X_{n,t}, -c_t(\Delta X_{n,t})\}= \text{\bfseries Select}(\{w_{n,t-1},X_{n,t-1}\},\rho_n(\cdot),c_t(\cdot))$ and send its trading request to the market maker
            \IF{trade happens}
                \STATE $X_{n,t} = X_{n,t-1} + \Delta X_{n,t}$
                \STATE $w_{n,t} = w_{n,t-1} - c_t(\Delta X_{n,t})$
            \ELSE
                \STATE $X_{n,t} = X_{n,t-1}, w_{n,t} = w_{n,t-1}, \Delta X_{n,t} = 0$
            \ENDIF
            \STATE $t = t + 1$
        \UNTIL{market is closed}
        \ENSURE $\{w_{n,t},X_{n,t}\}_{t=1,2,\ldots}, \{\Delta X_{n,t}\}_{t=1,2,\ldots}$
    \end{algorithmic}
\end{algorithm}

\subsection{Appropriate choice on the pricing rule $c_t(\cdot)$}
There has been plenty of work on studying the pricing rule $c_t(\cdot)$ of the market maker \citep{brahma2012,pennock2004}. A popular class of mechanisms is Hanson's market scoring rules \citep{hanson2007}. It is later formalised by \citet{abernethy2013}, who use a set of reasonable axioms to characterise the pricing mechanism. We apply their result to our framework.

Let $\Delta X_t \equiv \Delta X_{a_t,t}$ be the trade with the market maker at time $t$. Consider two situations:
\begin{inparaenum}[1)]
    \item a trade happens with the market maker in $\Delta X$; and
    \item a trade happens with the market maker in $\Delta X'$ and is \emph{immediately} followed by another trade $\Delta X''$, where $\Delta X = \Delta X' + \Delta X''$.
\end{inparaenum}
A natural requirement is that the cost of purchasing $\Delta X$ should be equal to the total cost of purchasing $\Delta X'$ and $\Delta X''$. If we accept this property, then we can always find a functional $c:\mathcal{X}\to\Real$ such that \citep{abernethy2013}
\begin{equation}\label{eq:pathindi}
    c_t(\Delta X_t) = c(\Delta X_1 + \cdots + \Delta X_{t-1} + \Delta X_t) - c(\Delta X_1 + \cdots + \Delta X_{t-1}).
\end{equation}
We say a pricing rule $c_t$ is \emph{path-independent} if it has the form of (\ref{eq:pathindi}), and reload the notation $c$ to represent $c_t$.

\section{The Machine Learning Objective of the Multi-period Trading Markets}\label{sec:objective}
Remember that the primary goal of this paper is to establish an intimate connection between machine learning and our new prediction markets model. Before we start to analyse the multi-period trading markets, we introduce the machine learning context for which we want our markets to be utilised. Many machine learning tasks could be interpreted under the following generic framework: given a set of data sampled from a space $\Omega$ and a hypothesis space $\mathcal{P}$ which contains a class of accessible probabilities on $\Omega$, we would like to find a probability from $\mathcal{P}$ that can best describe the data. Usually we use a functional $F:\mathcal{P}\to\Real$ to characterise the ``best'' performance, such that the best probability is the one that minimises $F$. Formally, this involves an optimisation problem
\begin{equation}
    \min_{P\in\mathcal{P}}F(P)
\end{equation}
For specific problems in which the information comes from different parts of the data or the models, $F$ has a form of $F = \sum_n F_n$, the sum of a set of functionals which share the same domain $\mathcal{P}$ (see examples in Section \ref{sec:examples} for details). We will show that a multi-period market effectively defines and optimises a machine learning task whose $F(P) = \sum_n F_n(P)$.

The connection is established in two steps: first we show that the market does have a global objective, and then show that under mild conditions the market optimises the dual of a machine learning problem $\min_{P\in\mathcal{P}}\sum_n F_n$.

\subsection{The global objective of a market}
We show that a multi-period trading market minimises a global objective. The optimisation is done sequentially via the market trading dynamics, that is, an agent will contribute to minimising this global objective as long as it makes a trade with the market maker. This argument is formalised in the following
\begin{proposition}[The global objective of a market]\label{thm:objective}
    A multi-period market (Algorithm \ref{alg:multiple}) with a path-independent pricing rule market maker aims to minimise the global objective
    \begin{equation}\label{eq:objective}
        L = c(Y) + \sum_{n\in A}\rho_n(X_n), \quad Y = \sum_{n\in A}X_n,
    \end{equation}
    by performing a sequential optimisation algorithm, which is implemented by the market trading process (cf.\ (\ref{eq:singlemin}) and (\ref{eq:update})):
    \begin{subequations}
        \begin{align}
            \Delta X_t &= \argmin\nolimits_{\Delta X'_t} \rho_{a_t}(X_{a_t,t-1}+\Delta X'_t +w_{a_t,t-1}-c_t(\Delta X'_t)), \label{eq:singleopt} \\
            X_{n,t} &= X_{n,t-1} + \ind(n=a_t)\Delta X_t, \\
            w_{n,t} &= w_{n,t-1} - \ind(n=a_t)c_t(\Delta X_t), \\
            Y_t &= Y_{t-1} + \Delta X_t,
        \end{align}
    \end{subequations}
    If the algorithm converges at time $t'$, i.e.\ $\Delta X_t = 0$ for all $t>t'$, then $\{X_{n,t'}\}, Y_{t'}$ achieves a local minimum of the objective $L$ in (\ref{eq:objective}).
\end{proposition}
\begin{proof}
    At time $t$ only agent $a_t$ will trade with the market maker, so $\Delta X_t = \Delta X_{a_t,t}$. At time $t$, for any agent $n$ all quantities calculated before $t$ can be treated as constants as they could no longer be modified. Therefore, the functional that is minimised in (\ref{eq:singleopt}) has the same optimal point with the following functional
    \begin{equation}
         l_t(\Delta X'_t) = \rho_{a_t}(X_{a_t,t-1}+\Delta X'_t+w_{a_t,t-1}-c_t(\Delta X'_t)) - \rho_{a_t}(X_{a_t,t-1}+w_{a_t,t-1}).
    \end{equation}
    Apply the property of translation invariance to $l_t$, we have
    \begin{equation}\label{eq:addconstants}
        l_t(\Delta X'_t) = \rho_{a_t}(X_{a_t,t-1}+\Delta X'_t) - \rho_{a_t}(X_{a_t,t-1}) + c_t(\Delta X'_t).
    \end{equation}
    Sum over all $l_t$'s and denote this summation by $L_T$, which is a functional. Then
    \begin{equation}\label{eq:sumover}
        \min_{\{\Delta X'_t\}} L_T = \min_{\{\Delta X'_t\}}\sum_{t=1}^Tl_t(\Delta X'_t) = \sum_{t=1}^T \min_{\Delta X'_t}l_t(\Delta X'_t) = \sum_{t=1}^T l_t(\Delta X_t).
    \end{equation}
    Here $\Delta X_t$'s are the optimal point obtained from (\ref{eq:singleopt}). Substitute (\ref{eq:addconstants}) to (\ref{eq:sumover})
    \begin{equation}
        \sum_{t=1}^T l_t(\Delta X_t) = \sum_{t=1}^T\rho_{a_t}(X_{a_t,t-1}+\Delta X_t) - \rho_{a_t}(X_{a_t,t-1}) + \sum_{t=1}^Tc_t(\Delta X_t).
    \end{equation}
    Note that at time $t$ for any agent $n\neq a_t$ it makes no trade $\Delta X_{n,t}=0$, and so
    \begin{equation}
        \rho_n(X_{n,t-1} + \Delta X_{n,t}) - \rho_n(X_{n,t-1}) = 0, \quad \forall n\neq a_t.
    \end{equation}
    The first summation on RHS thus becomes
    \begin{align}\label{eq:riskcancel}
        &\sum_{t=1}^T\rho_{a_t}(X_{a_t,t-1}+\Delta X_t) - \rho_{a_t}(X_{a_t,t-1}) \notag \\
            =& \sum_{t=1}^T\sum_{n\in A}\rho_n(X_{n,t-1}+\Delta X_{n,t}) - \rho_n(X_{n,t-1}) \notag \\
            =& \sum_{n\in A}\rho_n(X_{n,T}) - \rho_n(X_{n,0}).
    \end{align}
    Since the pricing rule is path-independent, the second summation on RHS is
    \begin{equation}\label{eq:pricingrulecancel}
        \sum_{t=1}^Tc_t(\Delta X_t) = \sum_{t=1}^Tc_t(Y_t) - c_t(Y_{t-1}) = c(Y_t) - c(0),
    \end{equation}
    where $Y_t = \sum_{\tau=1}^t\Delta X_\tau$ and $Y_0 = 0$. Since $X_{n,0} = 0$ and for any $t$ and $n\neq a_t$ $\Delta X_{n,t} = 0$, we have
    \begin{align}\label{eq:inventory}
        Y_t &= \sum_{\tau=1}^{t} \Delta X_\tau = \sum_{\tau=1}^t \Delta X_{a_\tau,\tau} = \sum_{\tau=1}^t\sum_{n\in A}\Delta X_{n,\tau} \notag \\
            &= \sum_{n\in A}\sum_{\tau=1}^t\Delta X_{n,\tau} = \sum_{n\in A} X_{n,t},   \quad \forall t>0.
    \end{align}
    Finally, substitute (\ref{eq:riskcancel}) (\ref{eq:pricingrulecancel}) and (\ref{eq:inventory}) to (\ref{eq:sumover}) and merge the rest terms we can end up with
    \begin{equation}
        \min_{\{\Delta X_t\}} L_T = \min_{\{\Delta X_t\}} c(Y_T) + \sum_{n\in A}\rho_n(X_{n,T}) - c(0),
    \end{equation}
    where $Y_T = \sum_{n\in A}\Delta X_{n,T}$. This is a sequential minimisation scheme for $\min L$. Finally, if the market converges at time $T$, we have $X_n = X_{n,T}$ and $Y = Y_T$, leading to at a local minimum of  $L$.
\end{proof}

Proposition \ref{thm:objective} is the key to understanding the market mechanism. Despite that the market is set up to let agents behave under their own preferences, the market mechanism ensures that a global objective is established, and that the agent will contribute to optimising the global objective at the same time as it optimise its own goal. The trading process thus provides a sensible algorithm for achieving this global objective.

\subsection{A primal-dual representation via convex analysis}
One concern is that (\ref{eq:objective}) is not commonly seen in machine learning problems\footnote{However, to complete our discussion, we show one example that uses (\ref{eq:objective}) in Section \ref{sec:examples}}. A different view of this objective should somehow be introduced. In fact, under mild requirements on the form of risk measures and pricing rules, the global objective forms the dual of the optimisation problem $\min_{P\in\mathcal{P}}\sum_n F_n(P)$. The requirement for the risk measures is convexity (\ref{eq:convexity}). The requirement for the pricing rules is that it is duality-based \citep{abernethy2013}.

\subsubsection{More on convex risk measures}
\citet{artzner1999} and \citet{follmer2002} show that a convex risk measure has a form
\begin{equation}\label{eq:dualrepresentation}
    \rho(X) = \sup_{Q\in\mathcal{P}} \left(\E\nolimits_Q[-X] - \alpha(Q) \right),
\end{equation}
where $\mathcal{P}$ is a set of probabilities on $(\Omega,\mathfrak{F})$ such that $Q$ is absolutely continuous w.r.t.\ $P$ and $\E_Q[X]$ is well defined. The risk measure decreases as $\E_Q[X]$ increases but this effect is penalised by a functional $\alpha$. (\ref{eq:dualrepresentation}) is in essence a Legendre-Fenchel transform with a slight change on signs \citep{boyd2004}.

\subsubsection{Duality-based pricing rules}
We keep following the idea of \citet{abernethy2013} and apply their duality-based pricing rules to our problem. The authors point out that duality-based pricing rules are well motivated as they meet some natural conditions such as no-arbitrage. A duality-based pricing rule is path-independent and has a form\footnote{\citet{abernethy2013} represent markets in securities $\{\xi_k\}$ and shares $\{s_k\}$. To be consistent with our framework we change the representation to assets $X$ (cf.\ Section \ref{sec:setup}).}
\begin{equation}\label{eq:dualitybasedpricing}
    c(X) \equiv \sup_{Q\in\mathcal{P}} \left(\E\nolimits_Q[X] - R(Q) \right) = R^*(X),
\end{equation}
where $R^*$ denotes the Legendre-Fenchel transform of $R$. Note that in their work $R$ is required to be convex, but this condition could be relaxed since for any $R$ we could define $R' \equiv (R^*)^* = c^*$ to replace $R$, as $R'$ is always convex (as it is a conjugate dual) and $c = (R')^* = R^*$.

\subsubsection{The primal problem}
Now we are ready to show
\begin{proposition}[The primal problem]\label{thm:duality}
    For a multi-period market which involves agents who use convex risk measures in (\ref{eq:dualrepresentation})and a duality-based pricing rule market maker in (\ref{eq:dualitybasedpricing}), its global objective is the dual of
    \begin{equation}\label{eq:primalrecall}
        \min_{P\in\mathcal{P}} \sum_{n=0}^N F_n(P),
    \end{equation}
    where $F_0$ and $F_n$ are functionals that share the same domain $\mathcal{P}$. Specifically, $F_0 = R$ in (\ref{eq:dualitybasedpricing}), and $F_n = \alpha_n$ where $\alpha_n$ is the penalty functional of agent $n$.
\end{proposition}
\begin{proof}
    We use the \emph{generalised Fenchel's duality} \citep{shalev-shwartz2007} to derive the Lagrange dual problem of (\ref{eq:primalrecall}). The generalised Fenchel's duality states that the dual of problem (\ref{eq:primalrecall}) is
    \begin{equation}\label{eq:dual}
        -\min_{\{X_n\}\in\mathcal{X}} F_0^*(Y) + \sum_{n=1}^NF_n^*(-X_n), \quad Y = \sum_{n=1}^NX_n,
    \end{equation}
    where $F_n^*$ denotes the Legendre-Fenchel transform.

    We construct the convex risk measure for each agent $n$. use (\ref{eq:dualrepresentation}) and choose $\alpha = F_n$
    \begin{equation}
        \rho_n(X) = \sup_{Q\in\mathcal{P}} \left(\E\nolimits_Q[-X] - F_n(Q) \right) = F_n^*(-X).
    \end{equation}
    For the pricing rule (\ref{eq:dualitybasedpricing}) we choose $R = F_0$ and obtain $c = F_0^*$. Substitute them back to the dual problem (\ref{eq:dual}) and we end up with
    \begin{equation}
        -\min_{\{X_n\}}L = -\min_{\{X_n\}\in\mathcal{X}} c(Y) + \sum_{n=1}^N\rho_n(X_n).
    \end{equation}
    This matches the global objective $L$ (cf. (\ref{eq:objective})) with a different sign. The negation sign is necessary because the Lagrange dual problem \emph{lower bounds} the primal one
    \begin{equation}\label{eq:lowerbound}
        -\min_{\{X_n\}}L \leq \min_{P\in\mathcal{P}} \sum_{n=0}^N F_n(P),
    \end{equation}
    and become exact when the strong duality holds \citep{boyd2004}.
\end{proof}
Proposition \ref{thm:duality} gives us two ways of building the connection between markets and machine learning:
\begin{inparaenum}[1)]
    \item If we model a market using our framework, we could then figure out the global objective of the market and then the primal problem, which can be solved using machine learning methods.
    \item More interestingly, given a machine learning problem of form (\ref{eq:primalrecall}), we could transform it to a market and solve the problem by running the market, during which we could take the advantage of some market properties, such as distributed environment and privacy, to gain extra benefits.
\end{inparaenum}

\section{Related Work}\label{sec:relatedwork}
The idea of building models for prediction markets and discussing their relation to optimisation is not novel, and a few significant progresses have been achieved in the past few years. We will discuss the work that is closely related to ours.

It is \citet{chen2010} who first show that what scoring rule market makers do are effectively online no-regret learning. Their study focuses on the market makers while agents are not directly modelled, which motivates a framework for the whole market.

\citet{storkey2011} defines and analyses a type of prediction markets based on definitions on the markets, securities, and agents. Agents are modelled by EUT, that is, an agent is rational by maximising its expected utility. By analysing the equilibrium status of the market the author shows that the market can aggregate beliefs from agents to output a probability distribution over the future events. The author does not discuss precise market mechanisms or give the global objective of the market, which makes it difficult to link these markets to optimisation procedures.

Another important progress is given by \citet{penna2012}, who apply the market scoring rules as the market mechanism to the framework of \citet{storkey2011}. The work shows that with a large population of agents whose portfolios are drawn from a demand distribution, the whole market implements \emph{stochastic mirror descent}. One concern is that they suggest using EUT to model agents but they do not use it to solve the optimal portfolios for the agents. This problem is partially solved by \citet{premachandra2013}, who derives the solution for a certain type of expected utilities. A similar setting is also studied by \citet{sethi2013}. They focus more on the convergence of the market dynamics, and show how markets can aggregate beliefs by using numerical evidences.

\subsection{Risk measures and EUT}
Here we justify the choice of risk measures as the agent decision rules. First, the output value of a risk measure can be treated as a risk-free asset and standard linear operations are well defined for it. In comparison, an expected utility outputs a number that only has abstract meaning, i.e.\ to measure the degree of agent's satisfaction. Additionally, risk measures force translation invariance by definition, but expected utility functions do not have this property in general. With the help of translation invariance, the wealth $w$ can always be separated from the risky asset $X$, which implies that the optimal portfolio does not depend on $w$. This saves us from the trouble of associating $w$ with the aggregation weights, as the relationship between them is highly inconsistent and varies dramatically under different utilities \citep{storkey2012}. Finally, we could always derive convex a risk measure $\rho_u$ from any expected utility \citep{follmer2004}
\begin{equation}
    \rho_u(X) \equiv \inf\{m\in\Real \mid \E\nolimits_P[u(X+m)]\geq u_0 \},
\end{equation}
where $P$ is the personal belief of the agent. In fact, the output of this risk measure is the \emph{risk premium}, the least amount of money that one would like to borrow in order to accept this risky asset. Then a sensible decision rule should be to find an asset that minimise the premium, which leads to our decision rule. 

As an example, consider the HARA utility
\begin{equation}
    u_H(x) = \frac{1-\gamma}{\gamma}\left(\frac{ax}{1-\gamma} + b\right)^\gamma, \quad a>0, \frac{ax}{1-\gamma} + b>0.
\end{equation}
The resultant convex risk measure is the one who has the following penalty functional
\begin{equation}\label{eq:harapenalty}
    \alpha(Q) = \frac{\gamma}{a}\eta^{-1/\eta}(-u_0)^{1/\gamma}\E\left[\left(\frac{\dif Q}{\dif P} \right)^\eta \right]^{1/\eta} + (1-\gamma)\frac{b}{a},
\end{equation}
where $1/\eta+1/\gamma = 1$. A special case of HARA is given by $b=1$ and $\gamma\to-\infty$, which leads to the exponential utility function $u_E(x)=-\exp(-ax)$. It is easy to check that the risk measure associated with exponential utility $u_E$ is exactly the entropic risk measure in (\ref{eq:entropic}) with $\theta=a$ \citep{follmer2002}.

\section{Examples}\label{sec:examples}
In this section we use three examples to illustrate the connections between the multi-period trading markets and machine learning.

\subsection{Opinion Pooling}
The opinion pooling problem is a common setting for prediction market models \citep{barbu2012,storkey2012}. \citet{garg2004} show that the objective of an opinion pool is to minimise a weighted sum of a set of divergences. Particularly, for logarithmic opinion pool the objective is to
\begin{equation}\label{eq:logop}
    \min_{P\in\mathcal{P}} \sum_nw_n D[P||P_n],
\end{equation}
where $D[\cdot||\cdot]$ is the KL-divergence and $\{w_n\}$ are weight parameters.

Now consider an log-opinion pool of a set of $A$ probabilities on a finite discrete sample space $\Omega$ with $K$ states. To set up a market that matches the log-opinion pool, we first define a market on the same space $\Omega$ and introduce $K$ Arrow-Debreu securities. We introduce $N$ agents, and assign a unique probability $P_n\in A$ to agent $n$ as its personal belief. According to (\ref{eq:entropic}), agent $n$'s risk measure has the form
\begin{equation}
    \rho_n(\sB_n) = \frac{1}{\theta_n}\log\sum_{k=1}^K p_ke^{-\theta_ns_{n,k}},
\end{equation}
where we let $\theta_n$ match the weight $w_n$ by $\theta_n^{-1} = w_n$. For the sake of simplicity, we choose a logarithmic market scoring rule market maker \citep{hanson2007}
\begin{equation}
    c(\sB_0) = \frac{1}{\theta_0}\log\sum_{k=1}^K e^{\theta_0s_{0,k}}.
\end{equation}
The market can be run by using Algorithm \ref{alg:multiple}. Two typical simulation results are shown in Figure \ref{fig:converged} and \ref{fig:stochastic}. The primal problem of this market is (applying Proposition \ref{thm:objective} and \ref{thm:duality})
\begin{equation}\label{eq:lopprimal}
    \min_{P\in\mathcal{P}} \frac{1}{\theta_0}D[P||P_0] + \sum_{n\in A} \frac{1}{\theta_n}D[P||P_n],
\end{equation}
where the domain $\mathcal{P} = \Delta_K$ is the probability simplex in $K$ dimensions and $P_0 = \mathrm{uniform}(K)$ is the discrete uniform distribution in $\Delta_K$. In this case the optimal $P$ can be analytically solved. Recall that $\theta_n^{-1} = w_n$ and we have
\begin{equation}\label{eq:biased}
    P\propto \prod_{n\in A}P_n^{w_n/(\theta_0^{-1} + \sum_{n\in A}w_n)}.
\end{equation}
Since we introduce the market maker, the aggregated belief $P$ is not a pure weighted product of agents' beliefs, but with a bias towards $P_0$. However, when the population is sufficiently large such that $\sum_n\theta_n^{-1}\gg \theta_0^{-1}$, the effect of the market maker could be ignored and we will end up with a pure aggregation of agent beliefs \citep{penna2012}.

\begin{figure}
    \centering
    \includegraphics[scale=.5]{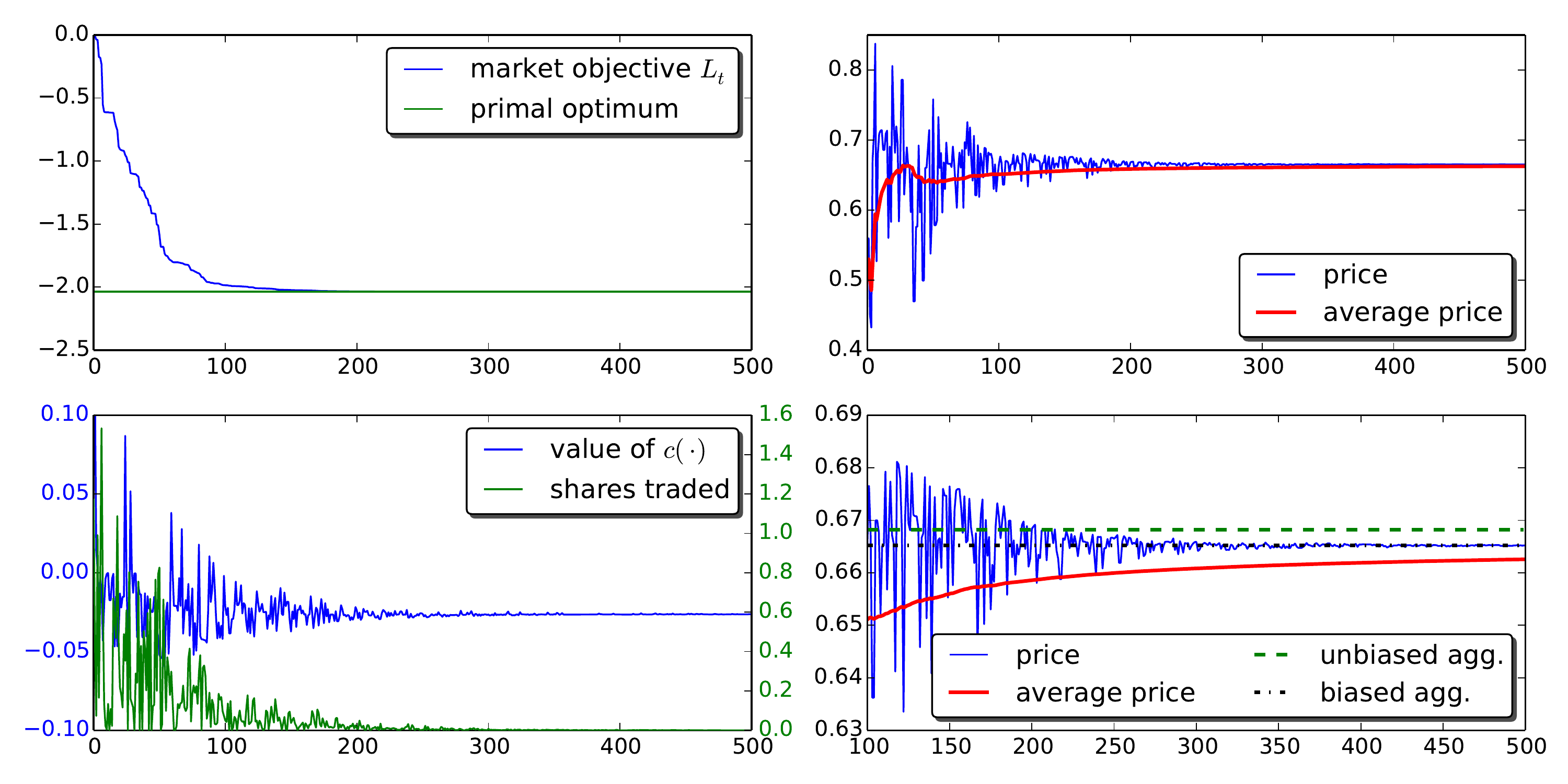}
    \caption{A market with Arrow-Debreu securities defined on a binary event $\omega$. The experiment setting is similar to the biased coin market \citep{penna2012}, in that all agents start with a uniform prior on $\omega$ and each one builds its own posterior belief after a private observation of $5$ samples of $\omega$. The difference is here only a finite number $N=10$ of agents are involved. After $t=300$ the global objective of the market converges to the negation of (\ref{eq:lopprimal}) (upper left, cf.\ (\ref{eq:lowerbound})), while the market price converges to a position which is close to the unbiased aggregation, but with a small bias towards $0.5$ due to the biased uniform belief of the market maker (upper and lower right, cf.\ (\ref{eq:biased})).}
    \label{fig:converged}
\end{figure}

\begin{figure}
    \centering
    \includegraphics[scale=.5]{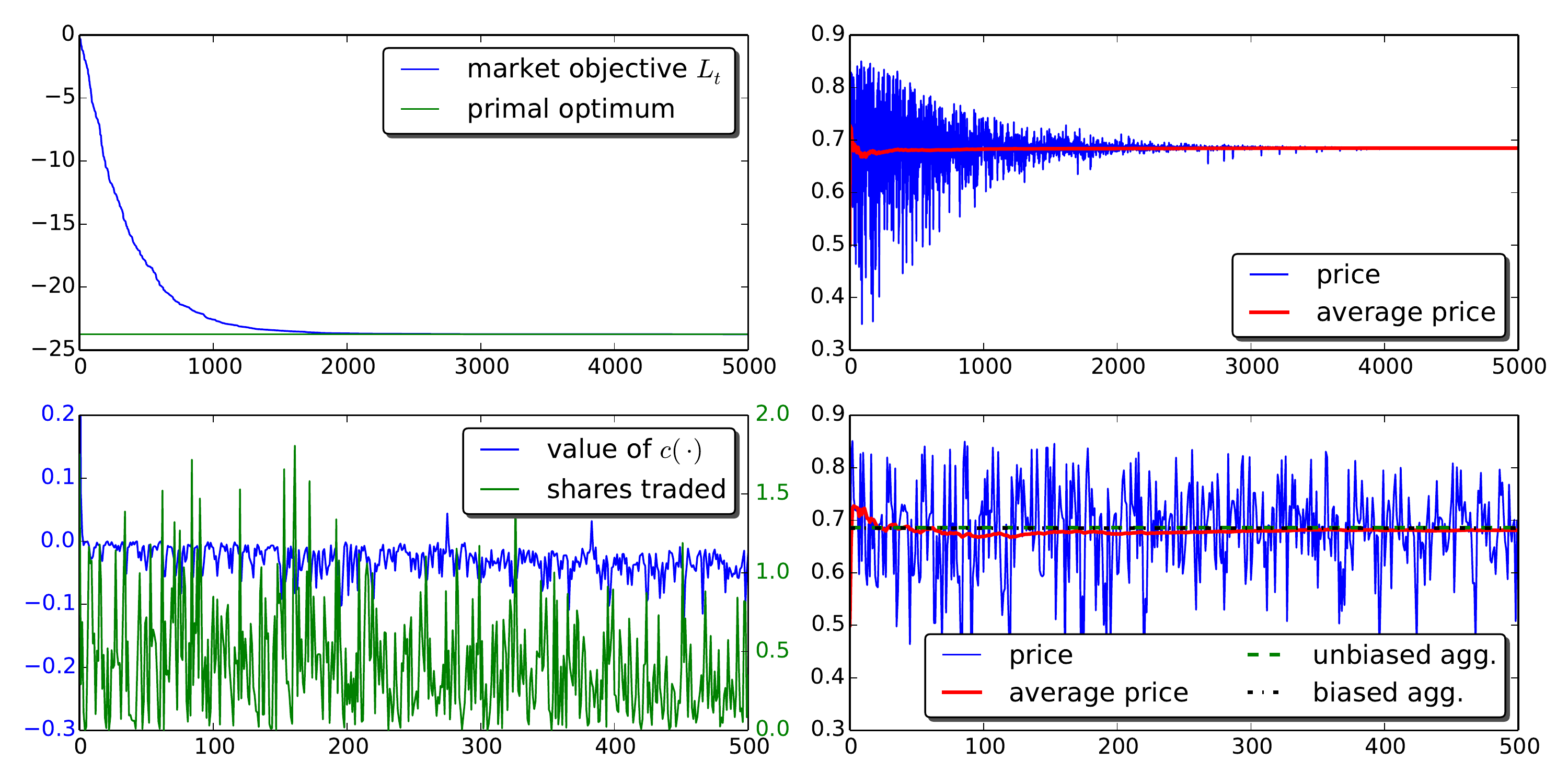}
    \caption{A market which shares the setting in Figure \ref{fig:converged} but with $N=100$ agents. After increasing the population the convergence becomes much slower (upper left and right). The market price does not show a sign of convergence before $t=500$ (lower left and right). Comparatively, the average price quickly converges to the aggregated belief. With more agents involved, the market maker loses its weight in the aggregation, leading to an aggregated belief that is closer to the unbiased one. This is expected, as for a large population the market should reproduce the results of the biased coin market \citep{penna2012}.}
    \label{fig:stochastic}
\end{figure}

\subsection{Bayesian Update}
We give our second example by first setting up a market and then match a machine learning problem to the market. Let's build a market on a continuous sample space $\Omega=\Real$. We only define one security $\xi(\omega) = \omega$, and so the asset $X=s\omega$. We introduce only one agent. Again, the agent is characterised by an entropic risk measures, with coefficient $\theta_1$ and $P_1=\mathcal{N}(\mu_1,\sigma_1^2)$ is the normal distribution. The moment-generating function in (\ref{eq:entropic}) is
\begin{equation}
    M_X(-\theta_1) = \E\nolimits_{P_1}[e^{-\theta_1s\omega}] = e^{-\theta_1s\mu_1+\sigma_1^2\theta_1^2s^2/2},
\end{equation}
and so the risk measure is $\rho_1(s) = -s\mu_1 + \sigma_1^2\theta_1s^2/2$. For the market maker we use the quadratic market scoring rule $c(s)=\theta_0s^2/2$. Now we could run this market using Algorithm \ref{alg:multiple} with only one agent.

It could be shown that this market implements a Bayesian \emph{maximum a posteriori} (MAP) update for the Gaussian, in which the prior is provided by the market maker and the likelihood information is provided by the agent.

Consider the setting of estimating a univariate Gaussian $\mathcal{N}(\mu,\sigma_1)$. All we need is the sufficient statistics calculated from a set of $N$ data points $\mathcal{D} = \{x_1,x_2,\ldots,x_N\}$. For clarity of exposition let's assume that we only care about the Bayesian updates of the mean parameter $\mu$, and think $\sigma_1$ is a prefixed constant. Introduce a Conjugate prior on the mean
\begin{equation}
    p(\mu\mid \mu_0,\sigma_0) \propto \exp\left(-\frac{1}{\theta_0}\frac{(\mu-\mu_0)^2}{2\sigma_0^2}\right),
\end{equation}
where $\theta_0^{-1}$ is so-called the pseudo count. The posterior is
\begin{align}
    p(\mu\mid \mathcal{D}&,\mu_0,\sigma_0) \propto p(\mu\mid \mu_0,\sigma_0) p(\mathcal{D}\mid\mu,\sigma_1) \notag \\
        &\propto \exp\left(-\frac{1}{\theta_0}\frac{(\mu-\mu_0)^2}{2\sigma_0^2}\right)\exp\left(-N\frac{(\mu-\bar{x})^2}{2\sigma_1^2}\right) \notag \\
        &\propto \exp\left(-\frac{1}{\theta_0}\frac{(\mu-\mu_0)^2}{2\sigma_0^2}-\frac{1}{\theta_1}\frac{(\mu-\mu_1)}{2\sigma_1^2} \right),
\end{align}
where $\mu_1=\bar{x}$ denotes the sample mean of the data set, and $\theta_1=N^{-1}$. If our goal is to calculate the MAP distribution then we have an optimisation problem
\begin{equation}
    L = \min_{\mu\in\Real} \frac{1}{\theta_0}\frac{(\mu-\mu_0)^2}{2\sigma_0^2} + \frac{1}{\theta_1}\frac{(\mu-\mu_1)^2}{2\sigma_1^2}.
\end{equation}
Let
\begin{equation}
    F_0(\mu) \equiv \frac{1}{\theta_0}\frac{(\mu-\mu_0)^2}{2\sigma_0^2}, \quad F_1(\mu) \equiv \frac{1}{\theta_1}\frac{(\mu-\mu_1)^2}{2\sigma_1^2},
\end{equation}
and thus we have $L = \min_{\mu\in\Real} F_0(\mu) + F_1(\mu)$. Since $F_0$ and $F_1$ are convex, we could apply the Fenchel's duality to the problem $L$, which gives us the following dual problem
\begin{equation}
    -L' = \min_{s\in\Real} F_0^*(s) + F_1^*(-s),
\end{equation}
where $F_0^*$ is the Legendre-Fenchel transform of $F_0$
\begin{equation}
    F_0^*(s) = \sup_{\mu\in\Real} s\mu - F_0(s) = s\mu_0 + \frac{1}{2}\sigma_0^2\theta_0s^2,
\end{equation}
and similarly $F_1^*(s) = s\mu_1 + \frac{1}{2}\sigma_1^2\theta_1s^2$. Choose the hyperparameter $\mu_0=0, \sigma_0=1$, and we finally have
\begin{equation}
    -L' = \min_{s\in\Real} \frac{\theta_0}{2}s^2 + \left( -s\mu_1 + \frac{1}{2}\sigma_1^2\theta_1s^2 \right) = \min_{s\in\Real} c(x) + \rho_1(s).
\end{equation}
This is exactly the agent's objective. Since $s$ and $\mu$ are dual to each other, the market performs the Bayesian update (MAP estimate) in the dual space of the mean parameters.

\subsection{Logistic Regression}
In the third example we discuss a classic machine learning problem. Given a data set $\mathcal{D} = \{\{\xB_m,\yB_m\}\mid \xB_m\in\Real^K,\yB_m = \{+1,-1\}, m=1,\ldots,M\}$, we would like to build logistic regression model with $l_2$-regularisation. The objective is
\begin{equation}\label{eq:logistic}
    L = \min_{\wB\in\Real^K}\frac{1}{M}\sum_{m=1}^M\log\left(1+e^{y_m(\wB\cdot\xB_m)}\right) + \frac{\lambda}{2}\|\wB\|^2,
\end{equation}
where $\|\cdot\|$ is the $l_2$ norm.

To convert this problem to a market we use (\ref{eq:objective}) and Proposition \ref{thm:objective}. Let the sample space be the space that generates the data $\Omega\equiv\Real^K\cup\{+1,-1\}$ and each future state is associated with a data in $\Omega$, $\omega=\{\xB,y\}$. Define $K$ securities, each of which is $\xi_k(\omega)=yx_k$. We introduce $N=K$ agents, such that the agent $n=k$ is only interested in trading in the $k$-th security $\xi_k$. Thus the shares of security $k$ held by agent $n$ is $s_{n,k} = \ind(n=k)w_k$, and the asset is $X_n = \sB_n\cdot\xiB = w_n\xi_n$. The market inventory is $\sB_0 = \sum_n\sB_n = \wB$. Let $c(\wB)$ be the first term on the RHS of (\ref{eq:logistic}) and define the risk measure of agent $n$ as $\rho_n(\sB_n) = \lambda\sB_n^2/2$. We end up with
\begin{equation}
    L = \min_{\wB} c(\wB) + \sum_{k=1}^K\frac{\lambda}{2}w_k^2 = \min_{\{\sB_n\}} c(\sB_0) + \sum_{n=1}^N\rho_n(\sB_n).
\end{equation}
Now the market is ready to run under Algorithm \ref{alg:multiple}. In order to show a slightly deeper connection to a specific learning method, we notice that the objective of agent $n$ at each round is $\min_{\Delta w_{k,t}}c(\wB_{t-1}+\Delta w_{k,t})+(w_{k,t-1}+\Delta w_{k,t})^2/2$. As the solution to this is not analytic, it is costly to solve for the exactly minimum of this objective at each time. To get rid of this problem, we could relax the condition that agents behaviour is rationally optimal, and let the agents accept a portfolio as long as it is better than its current position $\rho_n(\hat{\sB}_{n,t}) < \rho_n(\hat{\sB}_{n,t-1})$. Specifically agents can take steps towards the optimal solution. This can be achieved by the following portfolio updating rule
\begin{equation}
    \Delta w_{k,t} = -a \left.\frac{\dif}{\dif w_k}\left(c(\wB) + \frac{\lambda}{2}w_k^2 \right) \right|_{\wB=\wB_{t-1}},
\end{equation}
where $a>0$ is adjusted such that $\rho_n(\hat{\sB}_{n,t}) < \rho_n(\hat{\sB}_{n,t-1})$. In practice $a$ could be chosen by backtracking line search \citep{boyd2004}. The market we designed above effectively implements a \emph{coordinate descent} algorithm \citep{luo1992}.

Note that, instead of introducing $N=K$ agents, we can match the logistic regression problem by using only one agent and allowing it to trade all securities. This will result in a standard gradient descent method.

\section{Conclusion}
This paper establishes and discusses a new model for prediction markets. We use risk measures instead of expected utility to model agents, which results in an analytical market framework. We show that our market as a whole optimises certain global objective through its market dynamics. Based on this result, we make intimate connections between machine learning and markets.

One future work would be conducting a detailed analysis of this framework using the tools of convex optimisation \citep{boyd2004}. A particularly interesting topic is to find the conditions under which the market will converge. As we have observed, the stochasticity comes in when a large population of agents are involved, which is believed to be the nature of any real market \citep{penna2012}.

\section*{Acknowledgement}
This work was supported by Microsoft Research Cambridge through its PhD Scholarship Programme.

\bibliographystyle{apalike}
\bibliography{multiperiod}

\end{document}